\documentclass[11pt,letterpaper]{article}
\usepackage{amsfonts,amssymb,amsmath,latexsym,ae,aecompl,bbm,algorithm}
\usepackage{algorithm}
\usepackage[noend]{algpseudocode}
\usepackage{graphics}
\usepackage{graphicx}
\usepackage{color}
\usepackage{amsthm,enumitem}
\usepackage{fullpage}

\algblockdefx[ExecBl]{BlockOn}{BlockOff}  [1]{#1}
  
\makeatletter
\ifthenelse{\equal{\ALG@noend}{t}}%
  {\algtext*{BlockOff}}
  {}%
\makeatother

\newcommand{\remove}[1]{}

\usepackage{amsthm}

\makeatletter
\newtheorem*{rep@theorem}{\rep@title}
\newcommand{\newreptheorem}[2]{%
\newenvironment{rep#1}[1]{%
 \def\rep@title{#2 \ref{##1}}%
 \begin{rep@theorem}}%
 {\end{rep@theorem}}}
\makeatother

\newtheorem{theorem}{Theorem}
\newreptheorem{theorem}{Theorem}

\newtheorem{lemma}{Lemma}
\newreptheorem{lemma}{Lemma}

\newcommand{\INDState}[1][1]{\State\hspace{4ex}}

\usepackage{xparse}%

\usepackage{blindtext}

\usepackage{morewrites}

\usepackage{wrapfig}
\usepackage{xkeyval}%

\newwrite\authorbibfile%

\AtBeginDocument{%
	\immediate\openout\authorbibfile=\jobname.aub%
}%

\AtEndDocument{%
	\immediate\closeout\authorbibfile
	\InputIfFileExists{\jobname.aub}{}{}
}%

\makeatletter

\define@key{authorbib}{scale}[1]{%
	\def\AuthorbibKVMacroScale{#1}%
}

\define@key{authorbib}{wraplines}[10]{%
	\def\AuthorbibKVMacroWraplines{#1}%
}

\define@key{authorbib}{imagewidth}[4cm]{%
	\def\AuthorbibKVMacroImagewidth{#1}%
}

\define@key{authorbib}{overhang}[10pt]{%
	\def\AuthorbibKVMacroOverhang{#1}%
}

\define@key{authorbib}{imagepos}[l]{%
	\def\AuthorbibKVMacroImagepos{#1}%
}

\makeatother

\presetkeys{authorbib}{imagepos=l,imagewidth=4cm,wraplines=15,overhang=20pt}{}

\newlength{\AuthorbibTopSkip}
\newlength{\AuthorbibBottomSkip}
\NewDocumentCommand{\authorbibliography}{+o+m+m+m}{%
	\IfNoValueTF{#1}{%
	}{%
	\setkeys{authorbib}{#1}%
	\immediate\write\authorbibfile{%
		\string\begin{wrapfigure}[\AuthorbibKVMacroWraplines]{\AuthorbibKVMacroImagepos}[\AuthorbibKVMacroOverhang]{\AuthorbibKVMacroImagewidth}^^J
			\string\includegraphics[scale=\AuthorbibKVMacroScale]{#2}^^J
			\string\end{wrapfigure}^^J
	}%
}%
\IfNoValueTF{#3}{%
	\typeout{Warning: No author name}%
}{%
\immediate\write\authorbibfile{%
	\unexpanded{\vspace{\AuthorbibTopSkip}}^^J
	\string\noindent\relax
	\unexpanded{\textbf{#3}\par}^^J
	\string\noindent\relax
	\unexpanded{#4}^^J%
	\unexpanded{\vspace{\AuthorbibBottomSkip}}^^J
}%
}%
}%      

\begin{document}

\title{Achieving Dilution without Knowledge of Coordinates in the {SINR} Model}

\author{William K. Moses Jr.\thanks{Department of Computer Science and Engineering, Indian Institute of Technology Madras, Chennai, India. wkmjr3@gmail.com. This work was done while the author was an intern at Xerox Research Centre India, Bangalore.}
\and Shailesh Vaya~\thanks{Conduent Labs India, Bangalore, India. Shailesh.Vaya@conduent.com. This work was done while the author was with Xerox Research Centre India, Bangalore.}}

\date{}
\maketitle

\begin{abstract}
	When multiple radio nodes, situated in an arena, simultaneously transmit some messages, whether a node receives a message or not is determined by a mathematical formula involving the signal strengths of the transmitting nodes and their respective physical coordinates and sometimes even the topology in which they are located. Considerable literature has been developed for various fundamental distributed problems in this SINR (Signal-to-Interference-plus-Noise-Ratio) model for radio transmission. A setting typically studied is when all radio nodes transmit a signal of the same strength, and each device only has access to knowledge about the total number of nodes in the network $n$, the range from which each node's label is taken $[1,\dots,N]$, and the label of the device itself. In addition, an assumption is made that each node also knows its coordinates in the Euclidean plane. In this paper, we create a mechanism which allows algorithm designers to remove that last assumption in the given setting.  The assumption about the unavailability of the knowledge of the physical coordinates of the nodes truly captures the `ad-hoc' nature of wireless networks. From a practical and theoretical perspective, devising solutions for this setting is a very important challenge to undertake.
	
	Previous work in this area uses a flavor of a beautiful technique called dilution, in which the nodes transmit in a (predetermined) round-robin fashion, and are able to reach all their neighbors. However, without knowing the physical coordinates, it's not possible to know the coordinates of their containing (pivotal) grid box and seemingly not possible to use dilution (to coordinate their transmissions). We propose a new technique to achieve dilution without using the knowledge of physical coordinates. This technique exploits the understanding that the transmitting nodes lie in 2-D space, segmented by an appropriate pivotal grid, without explicitly referring to the actual physical coordinates of these nodes. Using this technique, it is possible for every weak device to successfully transmit its message to all of its neighbors in $\Theta(\lg N)$ rounds, as long as the density of transmitting nodes in any physical grid box is bounded by a known constant. This technique, we feel, is an important generic tool for devising practical protocols when physical coordinates of the nodes are not known.
\end{abstract}


\section{Introduction}\label{intro}
% !TEX root = ssf-dil-main.tex 
%please leave the above line untouched. Thanks, Billy.

The SINR (Signal-to-Interference-plus-Noise-Ratio) model for communication for ad-hoc wireless networks has been a theme of intensive, as well as extensive, study in recent years. It intimately models the hardness of communication in wireless networks, by explicitly considering the interference created by concurrently transmitting nodes and the ambient noise. Several of the deterministic protocols developed thus far have assumed that the physical coordinates of the nodes are known. However, this assumption severely restricts the deployment of ad-hoc wireless networks in the real world.

Assessing the physical coordinates and memorizing them is an arduous task (for e.g., applications are earthquake detection, defense and temperature detection in furnaces). Even if the initial locations of sensor nodes are known precisely\footnotemark[3]\footnotetext[3]{Highly accurate GPS with expensive iPhones are accurate only up-to $8$ meters.}, their coordinates change dynamically in most real world applications\footnotemark[4]\footnotetext[4]{Smart home with dozens of such devices working together to ensure that the home runs smoothly.}. Thus, from a very pragmatic deployment perspective, we would want to develop protocols that do not rely on this knowledge.

  In this work, we develop a tool to remove the reliance on this knowledge, assuming \textbf{weak connectivity} for SINR model formulated by Daum et al.~\cite{DGKN13} and further refined by Jurdzi\'nski and Kowalski~\cite{JK16}. We observe that we can successfully work with \textbf{weak devices}, as formulated by Jurdzi\'nski et al.~\cite{JKS12,JKS-ICALP-13}.

\subsection{Related Work}\label{previosresults}
  Previous deterministic protocols, in the domain, extensively use a technique called \textbf{dilution}, introduced by \cite{JKS12,JKS-ICALP-13}, for coordinating the nodes to successfully transmit their messages to their neighbors. This mechanism uses the knowledge of physical coordinates of the nodes to compute the coordinates of the pivotal grid box to which they belong. The coordinates of the grid box are then used to synchronize transmissions of the nodes. Dilution has been used to solve the problems of multi-broadcast \cite{JKS-ICALP-13,RKV15} and backbone creation \cite{JK12,RKV15}.

\subsection{Our Contributions}
 We develop a mechanism to achieve dilution without requiring the knowledge of the physical coordinates of the nodes. We call this new tool \textit{SSF Based Dilution}, as it uses strongly selective families from the ad-hoc wireless radio networks literature \cite{CMS01}. 

 First, consider that a grid $G_x$ is overlayed on the network. We generalize a proposition from \cite{JKS12} in order to show that, for a given node $u$, if we silence all nodes a sufficient distance away from it, all nodes within distance $\sqrt{2}x$ of $u$ receive its message. The resulting lemma is given below.
 
 \begin{replemma}{lem:ssf-dil}
 	For stations with same range $r$, sensitivity $\epsilon > 0$, and transmission power, for each $\alpha > 2$, there exists a constant $d$, which depends only on the parameters $\alpha$, $\beta$, and $\epsilon$ of the model and a constant $k$, satisfying the following property. 
 	
 	Let $W$ be the set of stations such that at most a constant $k$ of them want to transmit in any grid box of the grid $G_x$, $x \leq \frac{r}{\sqrt{2}}$. Let $u$ and $v$ be two stations in different grid boxes such that the distance between them, $\sqrt{2} x$, is the minimum distance between any two stations in different grid boxes in $G_x$. Let $A_u$ be the set of stations in $u$'s grid box.
 	
 	If $u$ is transmitting in a round $t$ and no other station within its box or a box less than $d$ box distance away from its box is transmitting in that round, then $v$ and all stations in $A_u$ can hear the message from $u$ in round $t$.
 \end{replemma}
 
 Our mechanism assumes that, in a given round, each grid box has at most a constant number of transmitting stations. We then use a strongly selective family to stretch this round over $O(\lg n)$ rounds in order to achieve an effect similar to dilution. This is captured in the following theorem.
 
 \begin{reptheorem}{the:ssf-replace}
 	For a grid $G_x$, $x \leq \frac{r}{\sqrt{2}}$, let the set of all nodes that want to transmit satisfy the properties of Lemma~\ref{lem:ssf-dil}. Every node in this set can successfully transmit a message to its neighbors within $\sqrt{2}x$ distance of it in $O(\lg N)$ rounds by executing one $(N,c)$-ssf, where $c = k^2(2d+1)^2$ where $d$ is taken from Lemma~\ref{lem:ssf-dil} and $k$, a constant, is an upper limit on the number of nodes from the set in any box of the grid.
 \end{reptheorem}
 
 Although we never know the actual physical coordinates of the nodes, we work with the \textit{imagination} that the radio nodes are embedded in 2-dimensional space.   The employment of this new mechanism of dilution requires protocols to reduce the setting to one where only a constant number of nodes are selected in each grid box for the purpose of transmission at the same time. 

 By removing the very strong assumption of the nodes knowing their actual physical coordinates, we feel that our result contributes to resolving an important challenge, which is to bridge the rich theoretical world of the SINR model and practical deployment of ad-hoc wireless networks in the real world.

\subsection{Organization of the paper}
  The rest of this paper is organized as follows. Section~\ref{sect:prelims} presents useful preliminaries. Section~\ref{sect:ssf-dil} presents the new technique of the ssf based dilution scheme, exploiting previous grid based frameworks in literature. Section~\ref{sect:conclusions} has our conclusions.

%---------------------------------------------------------------------------------------------------------
\section{Preliminaries}\label{sect:prelims}
% !TEX root = ssf-dil-main.tex 
%please leave the above line untouched. Thanks, Billy.

\paragraph{The SINR Model}
  In the SINR model, each wireless station, $u$, has some transmission power $P_u$ which is a positive real number. Let $d(*,*)$ be a distance function which gives the Euclidean metric distance between two stations. For a given round, let $\mathcal{T}$ be the set of all stations which are transmitting in that round. Now, the SINR for a station $u$'s message at station $v$ in that round is defined as follows:\\
$$SINR(u,v,\mathcal{T}) = \dfrac{\frac{P_u}{d(u,v)^{\alpha}}}{\mathcal{N} + \sum\limits_{i=\mathcal{T} \setminus u} \frac{P_i}{d(i,v)^{\alpha}}}.$$
  Here, $\alpha \geq 2$ is a fixed parameter of the model called the path loss constant. Note that we require $\alpha > 2$ for our technique. $\mathcal{N} \geq 0$ is also a fixed parameter of the model and refers to ambient noise. A node $v$ receives $u$'s message if and only if the SINR ratio of $u$'s message at $v$ crosses a threshold $\beta \geq 1$, which is also a fixed parameter of the model: 
  \begin{align}
  SINR(u,v,\mathcal{T}) \geq \beta. \label{eq:sinr-ineq}
  \end{align}

A device which only needs to satisfy Inequality~\ref{eq:sinr-ineq} for its message to be heard is called a \textbf{strong device}. If a device needs to further satisfy the following inequality, 
\begin{align}
\frac{P_u}{d(u,v)^{\alpha}} \geq (1 + \epsilon) \beta \mathcal{N}, \label{eq:weak-device-ineq}
\end{align}
where $\epsilon > 0$ is called the sensitivity parameter of the device, then it is called a \textbf{weak device}. Note that when $\epsilon = 0$, Inequality~\ref{eq:weak-device-ineq} reduces to Inequality~\ref{eq:sinr-ineq} (in the absence of interference). We make the assumption that all devices have the same fixed transmission power $P$.

A node $v$ \textbf{successfully receives} another node $u$'s transmission in a given round if both Inequalities~\ref{eq:sinr-ineq} and~\ref{eq:weak-device-ineq} are satisfied. The above definition is common in the literature, cf.~\cite{JKS12,KV10}. We define a station $u$'s \textbf{transmission range} as the maximum distance away at which another station can still successfully receive a message from $u$ in the absence of other transmitting stations. Since we assume that all nodes have the same power, it follows that they all have the same transmission range, $r$. A node $u$ \textbf{successfully transmits} in a given round if all nodes within range of $u$ successfully receive $u$'s message in the given round. Note that only if a node was a receiver in that round will it actually receive $u$'s message, else the message will be discarded.

  We define a \textbf{communication graph}, denoted by $G(V,E)$, as follows. Each station is considered as a node. If a node $v$ is within range of a node $u$, then we have an edge from $u$ to $v$. Since all nodes have the same range, if there exists an edge from $u$ to $v$, there also exists an edge from $v$ to $u$. We make the assumption that $G$ is connected. This construction of a communication graph based on nodes within range of each other and not just within a fraction of that range is called the \textbf{weak link} version of communication graphs \cite{DGKN13,JK16}. 
  
  Any algorithm prescribed for a station proceeds in a series of rounds, where each round corresponds to one global clock tick. In a given round, a station may act either as a receiver or a transmitter, but not both.
  Nodes do not have the ability to detect collisions. Nodes know the value of $n$, $N$, their own unique label, and a common ($N,c$)-ssf schedule (as defined below). Nodes don't know the value of their coordinates.

\paragraph{Grid and Pivotal Grid}
  The nodes are located on a Euclidean plane. We overlay a 2-dimensional grid $G_x$ on this plane such that the length of each side of a grid box is $x$ and a grid box is denoted by the coordinates of its bottom left coordinates. Therefore, a node with coordinates $(m,n)$ such that $ax \leq m < (a+1)x$ and $bx \leq n < (b+1)x$ has \textbf{grid coordinates} $(a,b)$ with respect to the grid $G_x$.
  The \textbf{pivotal grid} is a grid $G_{\frac{r}{\sqrt{2}}}$, and has been useful in designing algorithms \cite{DP07,EGKPPS09}. If two nodes are within the same grid box in $G_{\frac{r}{\sqrt{2}}}$, then they are within range of each other. The number of nodes located in a given box of the pivotal grid is not bounded.
  The idea of grid boxes within range of a given grid box refers to if a node can be located anywhere within a given grid box, what grid boxes can have nodes within range of that node. For a given grid box with grid coordinates $(a_1, b_1)$, we say another box with coordinates $(a_2, b_2)$ is at \textbf{box-distance} $0$ from it if the two boxes intersect, else at box distance $k$ from it where $k = max( min(|a_1 - a_2 - 1|, |a_2 - a_1 - 1|), min(|b_1 - b_2 - 1|, |b_2 - b_1 - 1|))$.
  
\paragraph{Dilution}
  A grid box is considered \textbf{active} in a given round if the node within that box is allowed to transmit in that round. Dilution, introduced by Jurdzi\'nski and Kowalski~\cite{JK12}, of a grid is the grouping of several grid boxes such that only one box of the group is active in any given round. Each grid box is active once in a cycle of rounds. Formally, a \textbf{$\delta$-dilution} of a grid $G_x$ implies that for a set of nodes, any two nodes with grid coordinates $(a_1,b_1)$ and $(a_2,b_2)$ will be active in the same round iff $(|a_1 - a_2| \mod \delta) = 0$ and $(|b_1 - b_2| \mod \delta) = 0$.
  
\paragraph{Strongly Selective Family}
  Let $N \geq c$ and both $N$ and $c$ be positive integers. An ($N,c$)-strongly selective family, commonly shortened to ($N,c$)-ssf, is a family $F$ of subsets of integers from $[1,N]$ such that for any non-empty integer subset $S$ of $[1,N]$, $|S| \leq c$, for each element $x \in S$, there exists a set $SS \in F$ such that $S \bigcap SS = {x}$. There exist $F$'s of size $O(c^2 \lg N)$\footnotemark[8]\footnotetext[8]{Previously in literature, it has been shown that this is an existential bound to the size of the strongly selective family. Recent work~\cite{BG15}, which is not yet published, shows a possible explicit construction of the desired $(N,c)$-strongly selective family in the form of an $(N, (1,c-1))$ cover free family.} which satisfy the above definition, cf. Clementi et al.~\cite{CMS01}.
  Let $S$ be an ($N,c$)-ssf and let $S_1, S_2, \ldots, S_o$ be sets belonging to $S$. A node \textbf{$u$ executes an ($N,c$)-ssf} when if $u \in S_i$, then in round $i$ of execution the node performs some action (e.g., transmission), and in other rounds it just acts as a receiver. The size of $S$ is $O(c^2 \lg N)$. Here, $c = k^2(2d+1)^2$, where $d$ is a constant and comes from Lemma~\ref{lem:ssf-dil}. For a given grid $G_x$, $x \leq \frac{r}{\sqrt{2}}$, $k$ is an upper limit on the number of nodes that want to transmit, present in any box of $G_x$. By ensuring that $k$ is a constant, the number of rounds of an ($N,c$)-ssf execution is $O(\lg N)$.

%---------------------------------------------------------------------------------------------------------
\section{Achieving Dilution without Knowledge of Coordinates}\label{sect:ssf-dil}
% !TEX root = ssf-dil-main.tex 
%please leave the above line untouched. Thanks, Billy.

  Dilution, formally introduced in \cite{JKS12}, helps to achieve the following hard task in SINR networks: Consider the grid $G_x$, $x \leq \frac{r}{\sqrt{2}}$, in which each grid box has at most one node which intends to transmit so that all its neighbors within distance $\sqrt{2} x$ successfully receive the message. If the nodes know their physical coordinates, then they know the coordinates of the grid box they belong to in the grid. Using this knowledge of coordinates of the grid box, dilution disables all nodes belonging to surrounding grid boxes around a node from transmitting, while allowing farther off nodes (limited in number) to transmit simultaneously. The net effect is that this node is able to successfully transmit to all its reachable neighbors. In fact, \cite{JKS12} showed, and extensively used, that all nodes can successfully transmit, by availing dilution, if it is given that there is at most one node that intends to transmit per grid box. The following lemma, called Proposition 11 in \cite{JKS12}, was their main tool to enable the success of dilution for weak devices.
  
\begin{lemma}[Proposition 11 in \cite{JKS12}]\label{lem:jks12-dilution}
For each $\alpha > 2$, there exists a constant $d$, which depends only on the parameters $\epsilon, \beta$ and $\alpha$ of the model, satisfying the following property. Let $W$ be a set of stations such that there is at most one station from $W$ in each box of the grid $G_x$, for some $x \leq \gamma$, and min$_{u,v \in W}\lbrace dist(u,v) \rbrace = x \cdot \sqrt{2}$. If station $u \in C$ for a box $C$ of $G_x$ is transmitting in a round $t$ and no other station in any box $C'$ of $G_x$ in the box-distance at most $d$ from $C$ is transmitting at that round, then $v$ can hear the message from $u$ at round $t$.
\end{lemma}
  
  Dilution is a very powerful tool, but it requires the nodes to know their physical coordinates. From a practical point of view, it's very hard for network installers to assess and record the physical coordinates of a node (and harder still if nodes are not merely in a 2-D Euclidean plane). This explicit knowledge severely restricts the ad-hoc nature of SINR networks. We develop a very interesting and powerful mechanism in this work by achieving dilution, quite surprisingly, without the knowledge of physical coordinates of the nodes. We are able to do this by generalizing Lemma~\ref{lem:jks12-dilution} as follows.

\begin{lemma}\label{lem:ssf-dil}
	For stations with same range $r$, sensitivity $\epsilon > 0$, and transmission power, for each $\alpha > 2$, there exists a constant $d$, which depends only on the parameters $\alpha$, $\beta$, and $\epsilon$ of the model and a constant $k$, satisfying the following property. 
	
	Let $W$ be the set of stations such that at most a constant $k$ of them want to transmit in any grid box of the grid $G_x$, $x \leq \frac{r}{\sqrt{2}}$. Let $u$ and $v$ be two stations in different grid boxes such that the distance between them, $\sqrt{2} x$, is the minimum distance between any two stations in different grid boxes in $G_x$. Let $A_u$ be the set of stations in $u$'s grid box.
	
	If $u$ is transmitting in a round $t$ and no other station within its box or a box less than $d$ box distance away from its box is transmitting in that round, then $v$ and all stations in $A_u$ can hear the message from $u$ in round $t$.
\end{lemma}

\begin{proof}
	This proof is along the lines of the proof of Proposition 11 in \cite{JKS12}. Consider two stations $u$ and $v$ satisfying the properties of the lemma where $u$ is trying to send a message to $v$. Note that the following argument goes through regardless of whether $u$ and $v$ are in the same box or in different boxes. We want to show that there exists a value of $d$ for each $\alpha$ such that if no other station within the same box as $u$ nor within box-distance less than $d$ from that box broadcasts, then $v$ successfully hears $u$'s message. We prove this by showing that for large enough $d$ for each $\alpha$, the noise generated by other broadcasting stations is insufficient to disrupt the message from reaching $v$.
	
	The strength of the signal transmitted by $u$ which is received at $v$, in round $t$, is $\frac{P}{(\sqrt{2} x)^\alpha}$.
	
	The number of boxes surrounding the box of $v$ at exactly box-distance $i$ from it is $8(i+1)$. The maximum number of stations that can transmit from those boxes is $8k(i+1)$. According to our assumptions, no stations in a box-distance of less than $d$ from $C$ can transmit in round $t$ while $u$ is transmitting. Therefore, the maximum interference that $v$ faces from other broadcasting stations and ambient noise is
	\begin{center}
		$\mathcal{N} + \sum\limits_{i=d}^\infty 8k(i+1) \cdot \dfrac{P}{\left( ix \right)^\alpha} = \mathcal{N} + \dfrac{8kP}{x^\alpha}\cdot c_d,$
	\end{center}
	where $c_d = \sum\limits_{i=d}^\infty \frac{1}{i^{\alpha - 1}} + \sum\limits_{i=d}^\infty \frac{1}{i^\alpha}$. Note that $c_d$ is the summation of two converging series, both dependent on $d$. $c_d$ can be made arbitrarily small by setting $d$ to a sufficiently large value. In order for $u$'s message to be heard by $v$, Inequality~\ref{eq:sinr-ineq} (SINR inequality) should be satisfied. Thus we have
	\begin{align}\label{eq:ssf-signal-beats-noise}
		(\mathcal{N} + \dfrac{8kP}{x^\alpha} c_d) \beta \leq \dfrac{P}{2^{\frac{\alpha}{2}} x^\alpha}.
	\end{align}
	
	Since $\sqrt{2} x \leq r$ and $r = \left( \frac{P}{(1 + \epsilon)\beta\mathcal{N}} \right)^{\frac{1}{\alpha}}$ (by Inequality~\ref{eq:weak-device-ineq} (weak device inequality) and the definition of $r$), we have $\dfrac{P}{\left( \sqrt{2} x \right)^\alpha} \geq (1 + \epsilon) \beta \mathcal{N}$. Therefore, we have a lower bound on the strength of the signal at $v$ of any node within range of $v$. If Inequality~\ref{eq:sinr-ineq} is satisfied for any such node's signal at $v$, then that automatically ensures that $u$'s message is heard by $v$.  Thus, we also have
	\begin{align}\label{eq:ssf-sinr-condition}
		(\mathcal{N} + \dfrac{8kP}{x^\alpha} c_d) \beta \leq (1 + \epsilon)\beta \mathcal{N}.
	\end{align}
	
	Our goal is for $u$'s message to be heard by $v$. If either Inequality~\ref{eq:ssf-signal-beats-noise} or \ref{eq:ssf-sinr-condition} is met, our goal is achieved. Simplifying Inequalities~\ref{eq:ssf-signal-beats-noise} and \ref{eq:ssf-sinr-condition}, we get the following condition imposed on $c_d$
	\begin{align}\label{eq:ssf-two-conditions}
		c_d \leq \dfrac{\frac{1}{\beta} - \frac{\mathcal{N} \left( \sqrt{2} x \right)^\alpha}{P} }{8 \cdot 2^{\frac{\alpha}{2}} \cdot k} \mbox{ or } c_d \leq \dfrac{ \epsilon \mathcal{N} x^\alpha }{8 k P}.
	\end{align}
	Our goal is to show that at least one of the inequalities is satisfied regardless of the values of $P$, $\mathcal{N}$, or $x$. Now depending on the value of $\frac{\mathcal{N} \left( \sqrt{2} x \right)^\alpha}{P}$, we have two cases to analyze:\\
	Case 1: $\frac{\mathcal{N} \left( \sqrt{2} x \right)^\alpha}{P} \leq \frac{1}{2 \beta}$\\
	The first inequality of Condition~\ref{eq:ssf-two-conditions} is reduced to requiring that $c_d \leq \frac{1}{16 \beta \cdot 2^{\frac{\alpha}{2}} \cdot k}$. For a given $\alpha$, there exists a sufficiently large value of $d$ which depends only on $\alpha$, $\beta$, and $k$ which satisfies the equation.\\
	Case 2: $\frac{\mathcal{N} \left( \sqrt{2} x \right)^\alpha}{P} > \frac{1}{2 \beta}$\\
	The second inequality of Condition~\ref{eq:ssf-two-conditions} is reduced to requiring that $c_d \leq \frac{\epsilon}{16 \beta \cdot 2^{\frac{\alpha}{2}} \cdot k}$. For a given $\alpha$, there exists a sufficiently large value of $d$ which depends only on $\alpha$, $\beta$, $\epsilon$, and $k$ which satisfies the equation.
	
	Thus, $v$ successfully hears $u$'s message in round $t$. Now consider all stations in the set $A_u$. The interference each of them experiences is less than that experienced by $v$. Thus, the above analysis holds good for each station in $A_u$ and each station in $A_u$ successfully hears $u$'s message in round $t$.
	%\qed
\end{proof}

  Using the above lemma, it is possible to achieve results similar to dilution using only an ($N,c$)-strongly selective family with an appropriately chosen constant $c$. For a grid $G_x$, $x \leq \frac{r}{\sqrt{2}}$, with at most one node per grid box who know their locations, $\delta$-dilution, where $\delta=(2d+1)^2$ where $d$ is taken from Proposition 11 in \cite{JKS12}, allows weak devices to successfully transmit their messages to their neighbors within distance $\sqrt{2} x$ of them. For the same grid, with at most \textbf{a constant number of nodes that want to transmit} per grid box who \textbf{don't know their locations}, executing an ($N,c$)-ssf schedule, where $c=k^2(2d+1)^2$, where $d$ is taken from Lemma~\ref{lem:ssf-dil} and $k$ is an upper limit on the number of nodes that want to transmit in any box of the grid, allows weak devices to successfully transmit their messages to their neighbors within distance $\sqrt{2} x$ of them. We call this approach a \emph{strongly selective family based dilution scheme}. 

The following theorem shows how we can replace any dilution scheme with our strongly selective family based dilution scheme when nodes satisfy the properties of Lemma~\ref{lem:ssf-dil}.

\begin{theorem}\label{the:ssf-replace}
For a grid $G_x$, $x \leq \frac{r}{\sqrt{2}}$, let the set of all nodes that want to transmit satisfy the properties of Lemma~\ref{lem:ssf-dil}. Every node in this set can successfully transmit a message to its neighbors within $\sqrt{2}x$ distance of it in $O(\lg N)$ rounds by executing one $(N,c)$-ssf, where $c = k^2(2d+1)^2$ where $d$ is taken from Lemma~\ref{lem:ssf-dil} and $k$, a constant, is an upper limit on the number of nodes from the set in any box of the grid.
\end{theorem}

\begin{proof}
From Lemma~\ref{lem:ssf-dil}, we get a $d$ value such that for a given node in box $C$, if all nodes within its box and all boxes less than box distance $d$ from it are silenced, that node will successfully transmit its message to its neighbors within $\sqrt{2} x$ distance from it. The number of boxes that need to be silenced is upper bounded by $(2d+1)^2$. Each of these boxes contain at most $k$ nodes that want to transmit. An ($N,c$)-ssf guarantees that for any set of size $c$ taken from $[N]$, there exists a subset such that only one element of those $c$ is chosen. Set $c = k^2(2d+1)^2$ and execute one ($N,c$)-ssf. For every node, there exists one round in the ($N,c$)-ssf such that all other nodes in its box and surrounding boxes are silenced and it can successfully transmit to all its neighbors within $\sqrt{2} x$ distance from it.
\end{proof}

  Using Theorem~\ref{the:ssf-replace}, it becomes possible for nodes to achieve results similar to dilution when they do not have knowledge of their grid coordinates and they satisfy the conditions of Lemma~\ref{lem:ssf-dil}.

 The following are some other salient features of the mechanism we feel we should highlight:
\begin{enumerate}
\item{} Unlike the mechanism in \cite{JKS12}, which was able to guarantee successful transmission to only a single node per box of a grid, we are able to guarantee successful transmission to as many nodes per box of the grid as needed, as long as this number is upper bounded by a known constant. Note that the restriction from Lemma~\ref{lem:ssf-dil} is on the number of nodes which are transmitting from each box. This is not the same as the number of nodes actually present in that box. By using a smart silencing mechanism and repeated use of an ($N,c$)-ssf, we can achieve dilution for larger grids when the number of nodes in a grid box is unbounded. The algorithm \emph{Multi-Broadcast} in \cite{MV16} is an example of such a process.
\item{} Our scheme does not require nodes to know their physical coordinates in the system, unlike \cite{JKS12}.
\item{} The identities of other nodes belonging to the same box of the grid, wanting to transmit, are not known to the other nodes in their box.
\item{} While it is guaranteed that there will be a round in which each node will be able to successfully transmit (as long as the constant density condition is met), it is not explicitly known which round it will be.
\item{} The new mechanism is somewhat \textit{fuzzy}, in the sense that while there is a round in which the (chosen) node will definitely transmit and be heard by all its neighbors, there may exist additional rounds in which the node transmits and is heard by a partial subset of its neighbors.
\end{enumerate}

\subsection{Example Applications of SSF Based Dilution}
The following are some examples of how our scheme may be applied.
\begin{enumerate}
	\item Consider the \emph{LeaderElection} algorithm in \cite{JKS12}. Lines 15 and 16 of that algorithm require the use of $\delta$-dilution for some constant $\delta$ and require nodes to know their locations. We can drop that assumption for that part of the algorithm and have nodes perform an ($N,c$)-ssf based dilution to achieve the same result with increase in running time by a factor of $O(\lg n)$ for lines 15 and 16.
	\item The wakeup problem specifies that $1 \leq k < n$ nodes are initially awake and we want to wake up the remaining nodes. The multi-broadcast problem requires that when $1 \leq k \leq n$ nodes each have a message, those messages be transmitted to all other nodes in the network. It is possible to develop algorithms to solve these two fundamental problems using ssf based dilution as seen in \cite{MV16}.
	\item A backbone is a constant degree connected dominating set (CDS) with asymptotically the same diameter as the network and such that every node in the network is connected to $\geq 1$ node in the backbone. Once a backbone is created on top of a network, it can be used to solve other communication problems such as routing using simple distributed algorithms related to the backbone. For the setting described in this paper, \cite{MV16} solve this problem by extensively using ssf based dilution. For the given setting, with the additional assumption that nodes know the labels of their neighbors, \cite{KMV16} solve this problem, again making extensive use of ssf based dilution. 
\end{enumerate}

%---------------------------------------------------------------------------------------------------------
\section{Conclusions}\label{sect:conclusions}
% !TEX root = ssf-dil-main.tex 
%please leave the above line untouched. Thanks, Billy.

  The important contribution of this work lies in devising a general purpose technique that achieves dilution (developed in \cite{JKS12}) without using knowledge of the actual physical coordinates of the wireless devices. Of course, we achieve our results by exploiting the understanding that the nodes are deployed in a reasonable metric and have associated with them actual physical coordinates but don't need to explicitly refer to them. This strongly selective families based notion is referred to as SSF Based Dilution. Potential applications of this technique include using it to realize communication tasks such as routing, multicast and group communication, and graph-kind of algorithms for clustering, coloring, MIS, CDS, etc. for SINR networks for consideration in the future.

\textbf{Acknowledgements:} We are very grateful to Darek Kowalski for several enriching discussions at various stages of this work.

%---------------------------------------------------------------------------------------------------------
\bibliographystyle{abbrv}
\bibliography{reference}

\begin{thebibliography}{10}

\bibitem{BG15}
N.~H. Bshouty and A.~Gabizon.
\newblock Almost optimal cover-free families.
\newblock {\em CoRR}, abs/1507.07368, 2015.

\bibitem{CMS01}
A.~E. Clementi, A.~Monti, and R.~Silvestri.
\newblock Selective families, superimposed codes, and broadcasting on unknown
  radio networks.
\newblock In {\em Proceedings of the twelfth annual ACM-SIAM symposium on
  Discrete algorithms}, pages 709--718. Society for Industrial and Applied
  Mathematics, 2001.

\bibitem{DGKN13}
S.~Daum, S.~Gilbert, F.~Kuhn, and C.~Newport.
\newblock Broadcast in the ad hoc sinr model.
\newblock In {\em Distributed Computing}, pages 358--372. Springer, 2013.

\bibitem{DP07}
A.~Dessmark and A.~Pelc.
\newblock Broadcasting in geometric radio networks.
\newblock {\em Journal of Discrete Algorithms}, 5(1):187--201, 2007.

\bibitem{EGKPPS09}
Y.~Emek, L.~G{\k{a}}sieniec, E.~Kantor, A.~Pelc, D.~Peleg, and C.~Su.
\newblock Broadcasting in udg radio networks with unknown topology.
\newblock {\em Distributed Computing}, 21(5):331--351, 2009.

\bibitem{JK12}
T.~Jurdzi\'nski and D.~R. Kowalski.
\newblock Distributed backbone structure for algorithms in the sinr model of
  wireless networks.
\newblock In {\em Distributed Computing}, pages 106--120. Springer, 2012.

\bibitem{JK16}
T.~Jurdzinski and D.~R. Kowalski.
\newblock Distributed randomized broadcasting in wireless networks under the
  {SINR} model.
\newblock In {\em Encyclopedia of Algorithms}, pages 577--580. 2016.

\bibitem{JKS12}
T.~Jurdzinski, D.~R. Kowalski, and G.~Stachowiak.
\newblock Distributed deterministic broadcasting in wireless networks under the
  {SINR} model.
\newblock {\em CoRR}, abs/1210.1804, 2012.

\bibitem{JKS-ICALP-13}
T.~Jurdzi\'nski, D.~R. Kowalski, and G.~Stachowiak.
\newblock Distributed deterministic broadcasting in wireless networks of weak
  devices.
\newblock In {\em Automata, Languages, and Programming}, pages 632--644.
  Springer, 2013.

\bibitem{KV10}
T.~Kesselheim and B.~V{\"o}cking.
\newblock Distributed contention resolution in wireless networks.
\newblock In {\em Distributed Computing}, pages 163--178. Springer, 2010.

\bibitem{KMV16}
D.~R. Kowalski, W.~K. Moses~Jr., and S.~Vaya.
\newblock Deterministic backbone creation in an {SINR} network without
  knowledge of location.
\newblock Manuscript.

\bibitem{MV16}
W.~K. Moses~Jr. and S.~Vaya.
\newblock Deterministic protocols in the {SINR} model without knowledge of
  coordinates.
\newblock Manuscript.

\bibitem{RKV15}
S.~P. Reddy, D.~R. Kowalski, and S.~Vaya.
\newblock Multi-broadcasting under the {SINR} model.
\newblock {\em CoRR}, abs/1504.01352, 2015.

\end{thebibliography}

\end{document}